\documentclass[11pt,a4paper]{article}

\usepackage{amssymb}
\usepackage{amsfonts}
\usepackage{amsmath}

\usepackage{multirow}
\usepackage[left=1.00in, right=1.00in, top=1.00in, bottom=1.00in]{geometry}

\usepackage{color-edits}

\addauthor{mb}{red}
\addauthor{lb}{blue}
\addauthor{nn}{green}

\newcommand{\qed}{\rule{1.5mm}{2mm}\vspace{0.1in}}
\newenvironment{proof}{\par\noindent{\bf Proof:}}{\qed}

\newcommand{\D}{\mathcal{D}}

\newcommand{\ignore}[1]{}
\newtheorem{theorem}{Theorem}[section]
\newtheorem{lemma}[theorem]{Lemma}

\newtheorem{proposition}[theorem]{Proposition}
\newtheorem{observation}[theorem]{Observation}
\newtheorem{claim}[theorem]{Claim}

\newtheorem{example}[theorem]{Example}

\begin{document}

\title{Selling Complementary Goods:\\ Dynamics, Efficiency and Revenue}
\author{Moshe Babaioff\footnote{Microsoft Research }  
$\;$ and
Liad Blumrosen\footnote{School of Business Administration, The Hebrew University } 
$\;$ and
Noam Nisan\footnote{School of Engineering and Computer Science, The Hebrew University and Microsoft Research}
}

\date{}
\maketitle

\begin{abstract}
We consider a price competition between two sellers of perfect-complement goods.
Each seller posts a price for the good it sells, but the demand is determined according to the sum of prices. This is a classic model by Cournot (1838), who showed that in this setting a monopoly that sells both goods is better for the society than two competing sellers.

We show that non-trivial pure Nash equilibria always exist in this game.
We also 
quantify Cournot's observation
with respect
to  both the optimal welfare and the monopoly revenue.
We then prove a series of mostly negative results regarding the convergence of best response dynamics to equilibria in such games.
\end{abstract}

\maketitle

\section{Introduction}


In this paper we study a model of a pricing game between two firms
that sell goods that are perfect complements to each other. These goods are only demanded in bundles, at equal quantities, and there is no demand for each good by itself.
The two sellers simultaneously choose prices $p_1, p_2$
and the demand at these prices is given by $\mathcal{D}(p_1+p_2)$ where
$\mathcal{D}$ is the demand for the bundle of these two complementary
goods. 
The revenue of seller $i$ is thus $p_i \cdot \mathcal{D}(p_1+p_2)$, and as
we assume zero production costs, this is taken as his utility.


This model was first studied in Cournot's famous work \cite{Cournot1838}.
In \cite{Cournot1838}, Cournot studied two seminal oligopoly models.
The first, and the more famous, model is the well known Cournot oligopoly model about
sellers who compete through quantities. We study a second model that was proposed by Cournot in the same work, regarding price competition between sellers of perfect complements.\footnote{
\cite{Son68} showed that these two different models by Cournot actually share the same formal structure.
}
In Cournot's example, a manufacturer of zinc may observe that some of her major customers produce brass (made of zinc and copper);
Therefore, zinc manufacturers
indirectly compete with manufacturers of copper, as both target the money of brass producers.
Another classic example of a duopoly selling perfect complements is by \cite{Ell39}, who studied how owners of two consecutive segments of a canal determine the tolls for shippers; Clearly, every shipper must purchase a permit from both owners for being granted the right to cross the canal.
Another, more contemporary, example for perfect complements might be
high-tech or pharmaceutical firms that must buy the rights to use two registered patents to manufacture its product;
The owners of the two patents quote prices for the usage rights, and these patents can be viewed as perfect complements. 

Cournot, in his 1838 book, proved a counterintuitive result saying that
competition among multiple sellers of complement goods lead to a \emph{worse} social outcome than the result reached by a monopoly that controls the two sellers. Moreover, both the profits of the firms and the consumer surplus increase in the monopoly outcome.
In the legal literature, this phenomenon was termed ``the tragedy of the anticommons" (see, \cite{BY00,Hel98,PSD05}). In our work, we will quantify the severity of this phenomenon.

%

Clearly, if the demand at a sufficiently high price is zero,
then there are {\em trivial} equilibria in which both sellers
price prohibitively high, and nothing is sold.
This raises the following question: {\em Do non-trivial equilibria, in which some pairs of items are sold, always exist? }
We study this question as well as some natural follow-ups:
{\em What are the revenue and welfare properties of such equilibria?
What are the properties of equilibria that might arise as a result of best-response dynamics?}




For the sake of quantification, we study a {\em discretized version} of this game in which the demand changes only finitely many times.
The number of discrete steps in the demand function, also viewed as the number of possible types of buyers,
is denoted by $n$ and is called the number of {\em demand levels}.

Our first result proves the existence of non-trivial pure Nash equilibria.

\begin{theorem}
	For any demand function with $n$ demand levels there exists at least one non-trivial pure Nash equilibrium.
\end{theorem}
We prove the theorem using an artificial dynamics which starts from zero prices and continues in steps.
In each one of these steps one seller best responds to the other seller's price, 
and after each seller best responds, the total price of both is symmetrized: both prices are replaced by their average.
We show that 
	the total price is monotonically non-decreasing, and thus it
terminates after at most $n$ steps in the non-trivial equilibrium of highest revenue and welfare.

In our model, it is easy to observe that there are multiple equilibria for some demand functions. How different can the welfare and revenue of these equilibria be?
A useful parameter for bounding the difference, as well as bounding the inefficiency of equilibria, turns out to be $D$,
the ratio between the demand at price $0$ and the demand at the highest price $v_{max}$
for which there is non-zero demand. 

Consider the following example with two ($n=2$) types of buyers:
a single buyer that is willing to pay ``a lot'', $2$,
for the bundle of the two goods, and many,
$D-1>>2$, buyers that are willing to pay ``a little'', $1$, each, for the bundle.
A monopolist (that controls both sellers) would have sold
the bundle at the low price $1$. At this price, all the $D$ buyers decide to buy, leading to revenue $D$ and optimal social welfare of $D+1$.
Equilibria here belong to two types:
the ``bad'' equilibria\footnote{It turns out that in our model there is no conflict between welfare and revenue in equilibria - the lower the total price, the higher the welfare and the revenue in equilibria (see Proposition \ref{obs:best-NE-well-defined}).}
have high prices, $p_1+p_2=2$, (which certainly is
an equilibrium when, say, $p_1 = p_2 = 1$) and achieve low revenue and
low social welfare of $2$.
The ``good'' equilibria have low prices, $p_1+p_2=1$ (which is an equilibrium as long
as $p_1, p_2 \ge 1/D$), and achieve optimal social welfare as well as the monopolist revenue, both values are at least $D$.
Thus, we see that the ratio of welfare (and revenue) between the ``good'' and ``bad'' equilibria can be very high, as high as $\Omega(D)$.
This can be viewed as a negative ``Price of Anarchy''  result.  

We next focus on the best equilibria and  present bounds on the ``Price of Stability" of this game;
	We show that the ratio between the optimal social welfare and the best equilibrium revenue\footnote{Note that this also shows  the same bounds on the ratio between the optimal welfare and the welfare in the best equilibrium, as well as the ratio between the monopolist revenue and the revenue in the best equilibrium.}
	is bounded by $O(\sqrt{D})$, and that this is tight when $D=n$.
	When $n$ is very small, the ratio can only grow as $2^n$ and not more. In particular, for constant $n$ the ratio is a constant, in contrast to the lower bound of $\Omega(D)$ on ``Price of Anarchy'' for $n=2$, presented above.
	


%

\begin{theorem}	
	For any instance, the optimal welfare and the monopolist revenue are at most
	$O(\min \{2^n,\sqrt{D}\})$ times the revenue of the best equilibrium. These bounds are tight.

	
\end{theorem}

\ignore{ 
	
In the above example we saw that the ratio between the welfare of the best and {\em worst} non-trivial equilibria can be as large as $\Omega(D)$ even with only two demand levels.
\mbedit{In contrast, the theorem shows that 
the ratio between  the {\em optimal welfare} and the best equilibrium welfare is only a constant when the number of demand levels is a constant, and that it cannot grow faster than $O(\sqrt{D})$ when $n\leq \sqrt{D}$.}

} 


We now turn to discuss how such markets converge to equilibria, and in case of multiple equilibria, which of them will be reached?
We consider best response dynamics in which
players start with some initial prices 
and repeatedly best-reply to each other.
 We study the quality of equilibria reached by the dynamics, compared to the best equilibria.

Clearly, if the dynamics happen to start at an equilibrium, best replying will leave the prices there,
whether the equilibrium is good or bad.
But what happens in general:
which equilibrium will they ``converge'' to when starting from ``natural" starting points, if any,
and how long can that take?
Zero  prices (or other, very low prices) are probably the most natural starting point.
However, as can be seen by the example above, starting from zero prices may result in the worst equilibrium.\footnote{In this example, the best response to price of $0$ is price of $1$. Next, the first seller will move from price of $0$ to price of $1$ as well, resulting in the worst equilibrium.}
Another natural starting point is a situation where the two sellers form a cartel and decide to post prices that sum to the monopoly price.
Indeed, in our example above, if the two sellers equally split the monopoly price, this will be the best equilibrium.
However, we know that cartel solutions are typically unstable, and the participants will have incentives to deviate to other prices and thus start a price updating process. We prove a negative result in this context, showing that starting from any split of the monopoly price might result in bad equilibria.
We also check what would be the result of dynamics that start at random prices. Again, we prove a negative result showing situations where dynamics starting from random prices almost surely converge to bad equilibria.
Finally, we show that convergence might take a long time, even with only two demand levels. Following is a more formal description of these results about the best-response dynamics:
%

\begin{theorem}
	The following statements hold:
	\begin{itemize}
		\item There are instances with $3$ demand levels for which a best-response dynamics starting from any split of a monopoly price reaches the worst equilibrium that
		is factor $\Omega(\sqrt{D})$ worse than the best equilibrium in terms of both revenue and welfare.
		
		\item For any $\epsilon>0$ and $D>2/\epsilon$  there are instances with $2$ demand levels for which a best-response dynamics starting from uniform random prices in $[0,v_{max}]^2$ reaches the worst equilibrium with probability $1-\epsilon$, while the best equilibrium has welfare and revenue that is factor $\epsilon\cdot D$ larger.
		
		\item For any $n\geq 2$ and $\epsilon>0$ there are instances with $n$ demand levels for which a best-response dynamics starting from uniform random prices in $[0,v_{max}]^2$ almost surely (with probability 1) reaches the worst equilibrium, while the best equilibrium has welfare and revenue that is factor $\Omega(2^n)$ larger.
			
		\item (Slow convergence.) For any $K>0$ there is an instance with only $2$ demand levels ($n=2$) and $D<2$ for which a best-response dynamics continues for at least $K$ steps before converging to an equilibrium.
	\end{itemize}
\end{theorem}

Thus, best-reply dynamics may take a very long time to converge, and then typically end up at a very bad equilibrium.
While for very simple ($n=2$) markets we know that convergence will always occur, we do not know whether convergence
is assured for every market.

\vspace{0.1in}
\noindent
{\bf Open Problem:} {\em Do best reply dynamics always converge to an equilibrium or may
they loop infinitely? }  We do not know the answer even for $n=3$.
\vspace{0.1in}

\noindent \textbf{More related work.}
While this paper studies price competition between sellers of perfect complements, the classic Bertrand competition \cite{Bert83} studied a similar situation between sellers of perfect substitutes. Bertrand competition leads to an efficient outcome with zero profits for the sellers.
\cite{BLN13} studied Bertrand-like competition over a network of sellers.
In another paper \cite{BBN16}, we studied a network of sellers
of perfect complements, where we showed how equilibrium properties depend on the graph structure, and we proved price-of-stability results for lines, cycles, trees etc.
Chawla and Roughgarden \cite{CR08} studied the price of anarchy in two-sided markets with consumers interested in buying flows in a graph from multiple sellers, each selling limited bandwidth on a single edge. Their model is fundamentally different than ours (e.g., they consider combinatorial demand by buyers, and sellers with limited capacities) and their PoA results are with respect to unrestricted Nash Equilibrium, while we focus on non-trivial ones (in our model the analysis of PoA is straightforward for unrestricted NE). A similar model was also studied in \cite{CN09}.

\cite{ES92} extended the complements model of Cournot to accommodate multiple brands of compatible goods. \cite{EK06} studied pricing strategies for complementary software products. The paper by \cite{FK01} directly studied the Cournot/Ellet model, but when buyers approach the sellers (or the tollbooths on the canal) sequentially.

\cite{FKLMO13} discussed best-response dynamics in a Cournot Oligopoly model with linear demand functions, and proved that they converge to equilibria.
Another recent paper \cite{NP10} studied how no-regret strategies converge to Nash equilibria in Cournot and Bertrand oligopoly settings;
The main results in \cite{NP10} are positive, showing how such strategies lead to a positive-payoff outcomes in Bertrand competition, but they do not consider such a model with complement items.

Best-response dynamics is a natural description of how decentralized markets converge to equilibria, see, e.g., \cite{FPT04,NSVZ11}, or to approximate equilibria, e.g., \cite{AAEMS08,SV08}.
The inefficiency of equilibria in various settings has been extensively studied, see, \cite{KP99, RT02, Rou15, ADKTWR08,HHT14}.

\vspace{2mm}

We continue as follows: Section \ref{sec:model} defines our model and some basic equilibrium properties. In Section \ref{sec:existence} we prove the existence of non trivial equilibria. In Section \ref{sec:best-response} we study the results of best-response dynamics. Finally, Section \ref{sec:pos} compares the quality of the best equilibria to the optimal outcomes.


\section{Model and Preliminaries}
\label{sec:model}

We consider two sellers, each selling a single, homogeneous, divisible good.
The sellers have zero manufacturing cost for the good they sell, and an unlimited supply is available from each good.
All the buyers in the economy are interested in bundles of these two goods, and the goods are perfect complements for the buyers.
That is, each buyer only demands a bundle consists of two goods,
in equal quantities\footnote{This actually assumes that the ratio of demand of the two goods is fixed, as we can normalized the units to assume that it is $1$ for both.},
and there is no demand for each good separately.
The demand for the bundle of the two goods is given by
a demand function $\mathcal{D}(\cdot)$, where $\mathcal{D}(p)\in \mathbb{R}_+$ is the quantity
of each of the two goods which is demanded when the price for one unit of the bundle of the two goods is $p\in \mathbb{R}_+$.

The sellers simultaneously offer prices for the goods they sell. Each seller offers a single price, and cannot discriminate between buyers.
If the two prices offered by the sellers are $p$ and $q$ then  $p+q$ is the {\it total price} and the demand in this market is $\mathcal{D}(p+q)$.
The revenue of the seller that posts a price $p$ is thus $p\cdot \mathcal{D}(p+q)$, the revenue of the second seller is $q\cdot \mathcal{D}(p+q)$ and the total revenue
of the two selling firms is denoted by $R(p+q)= (p+q)\cdot \mathcal{D}(p+q)$.
The maximal revenue that a monopoly that owns the two sellers can achieve
is $\sup_x x\cdot \mathcal{D}(x)$ and we use $p^*$ to denote a monopolist price.\footnote{
Our paper considers demand functions for which the monopoly revenue is attained and a monopolist price exists.	
When there is more than one price that maximizes the monopoly profit, our claims regarding $p^*$ will hold for each one of these prices. When necessary, we will treat the different prices separately.
}

\vspace{2mm}

\noindent \textbf{Discrete Demand Levels:}
In this paper we consider discrete demand curves, where potential buyers only have $n\geq 2$ different values denoted by $\vec{v}$, such that $v_1 > v_2 > \cdots > v_n>0$.
The demand at each price $v_i$ is denoted by $d_i=\mathcal{D}(v_i)$, and assuming a downward sloping demand curve we get that $\vec{d}$ is increasing, that is,
$0<d_1<d_2<\cdots<d_n$.
For convenience, we define $v_0=\infty$ and $d_0=0$.
The parameter $n$ is central in our analysis and
it denotes the number of \emph{demand levels} in the economy.
Another parameter that we frequently use is the
\emph{total demand} $D$, which is the ratio between the highest and lowest demand at non-zero prices, that is $D=d_n/d_1$.
In other words, $D$ is
the maximal demand $d_n$ measured in units of the minimal non-zero demand $d_1$ (Note that $D>1$).
The \emph{social welfare} in the economy is the total value generated for the consumers. The social welfare, given a total price $x$, is
 $SW(x)= \sum_{i | x<v_{i}} v_i (d_{i}-d_{i-1})$, and the optimal welfare is $SW(0)=\sum_{i=1}^{n} v_i (d_{i}-d_{i-1})$.

\vspace{2mm}

\noindent \textbf{Strategies and Equilibria:}
The sellers engage in a price competition.
We say that $p$ is a \emph{best response} to a price $q$ of the other seller if $p\in  argmax_{p'} \  p'\cdot \mathcal{D}(p'+q)$, and let the set of all best responses to $q$ be $BR(q)$.
We consider the pure Nash equilibria (NE) of this full-information pricing game.
A pure Nash equilibrium is a pair of prices such that each price is a best response to the other price, that is, $(p,q)$  such that  $p\in BR(q)$ and $q\in BR(p)$.

It is easy to see that NE always exist in this game, but unfortunately some of them are trivial and no item is sold, and thus their welfare is zero;
For example, $(\infty,\infty)$ is always an equilibrium with zero welfare and revenue.
We will therefore focus on a subset of NE that are {\em non-trivial}, i.e., where some quantity is sold.
It is not immediate to see that non-trivial 
equilibria exist, and we will begin by proving (in Section \ref{sec:existence}) that such equilibria indeed always exist.
On the other hand, we will show that
multiplicity of equilibria is a problem even for this restricted set of equilibria, as there might be an extreme variance in their revenue and efficiency.

\subsection{Basic Equilibrium Properties}

We now describe some basic structural properties of equilibria in the pricing game between sellers of complement goods. We use these properties throughout the paper.

We start with a simple observation claiming that all best response dynamics lead to a total price which is exactly one of the demand values.
This holds as otherwise any seller can slightly increase his price, selling the same quantity and increasing his revenue.

\begin{observation}
	\label{obs:BR-price-is-value}
	Let $x\leq v_1 $ be some price offered by one seller, and $BR(x)$ be a best response of the other seller to the price $x$. Then, it holds that $x+BR(x)=v_i$ for some $i\in \{1,...,n\}$.
	In particular, for every pure non-trivial NE $(p,q)$, it holds that $p+q=v_i$ for some $i$.
\end{observation}

Next, we prove a useful lemma claiming that the set of equilibria with a particular total price is convex.
Intuitively, the idea in the proof is that a seller with a higher offer cares more about changes in the demand than a seller with a lower offer. Therefore, if the seller with the higher offer decided not to deviate to an increased price, clearly the other seller would not deviate as well. The proof of the lemma appears in appendix \ref{app:eq-are-convex}.
	
\begin{lemma}
	\label{obs:total-price-convex}
If $(p,q)$ is a pure NE then $(x, p+q-x)$ is also a pure NE for every
$x\in [\min \{p,q\}, \max \{p,q\}]$.
	In particular, $((p+q)/2, (p+q)/2)$ is also a pure NE.
\end{lemma}

We next observe that there is no conflict between welfare and revenue in equilibrium:
an equilibrium with the highest welfare also has the highest equilibrium revenue. This holds
since
equilibria with lower total price obtain higher revenue and welfare.
We can thus say that any equilibrium with minimal total price is the ``best" as it is as good as possible on both dimensions: welfare and revenue. Similarly,  any equilibrium with maximal total price is the ``worst".
\begin{proposition}
\label{obs:best-NE-well-defined}
	Both welfare and revenue of equilibria are monotonically non-increasing in the total price.
	Therefore, an equilibrium with the minimal total price has both the highest welfare and the highest revenue, among all equilibria.
	Similarly, an equilibrium with the maximal total price has both the lowest welfare as well as the lowest revenue, among all equilibria.
\end{proposition}
\begin{proof}
	Consider two equilibria, one with total price $v$ and the other with total price $w>v$. The claim that the welfare is non-increasing in the total price follows immediately from the definition.
We will show that for $w>v$ it holds that $R(v)\geq R(w)$.
	
	
	Lemma \ref{obs:total-price-convex} shows that if there is an equilibrium with total price $p$ then $(p/2,p/2)$ is also an equilibrium. As $(v/2,v/2)$ is an equilibrium, it holds that deviating to $w-v/2$ is not beneficial for a seller, and thus
	$R(v)/2\geq (w-v/2)\mathcal{D}(w)\geq (w/2) \mathcal{D}(w)= R(w)/2$ and thus $R(v)\geq R(w)$ as claimed.
\end{proof}

Finally, we give a variant of a classic result by Cournot \cite{Cournot1838}, which shows, somewhat counterintuitively, that
a single monopolist that sells two complementary goods is better for the society than two competing sellers for each selling one of the good. 
\begin{proposition}
	The total price in any equilibrium is at least as high as the minimal monopolist price $p^*$.
	Thus, the welfare and revenue achieved by the monopolist price $p^*$
	are at least as high as the welfare and revenue (resp.) of the best equilibria.
\end{proposition}
\begin{proof}
	Assume that there is an equilibrium with total price $p<p^*$. As $p^*$ is the minimal monopolist price it holds that $R(p)<R(p^*)$. Additionally, as there is an equilibrium with total price $p$ then by Lemma \ref{obs:total-price-convex} the pair $(p/2,p/2)$ is  an equilibrium, where each seller has revenue $R(p)/2$.
	As $p<p^*$ a seller might deviate to $p^*-p/2>p^*/2>0$, and since such deviation is not beneficial, it holds that
	$R(p)/2 \geq (p^*-p/2) \mathcal{D}(p^*)> (p^*/2) \mathcal{D}(p^*) = R(p^*)/2$ and thus $R(p)>R(p^*)$, a contradiction.

By Proposition \ref{obs:best-NE-well-defined}, it follows that the welfare and revenue achieved by the minimal monopolist price $p^*$ are no less than those in the best equilibrium.
\end{proof}

\section{Existence of Non-Trivial Equilibria}	
\label{sec:existence}


In this section we show that non-trivial equilibria always exist.
We first note that the structural lemmas from the previous sections seem to get us almost there:
We know from Obs. \ref{obs:BR-price-is-value} that the total price in equilibrium must equal one of the $v_i$'s; We also know that if $p,q$ is an equilibrium, then $(\frac{p+q}{2},\frac{p+q}{2})$ is also an equilibrium.
Therefore, if an equilibrium exists, then $(\frac{v_i}{2},\frac{v_i}{2})$ must be an equilibrium for some $i$.
However, these observations give a simple way of finding an equilibrium {\em if an equilibrium indeed exists}, but  they do not prove existence on their own.

We give a constructive existence proof, by showing an algorithm based on an artificial dynamics that always terminates in a non-trivial equilibrium. The algorithm is essentially a sequence of best responses by the sellers, but with a twist: after every best-response step the prices are averaged.
We show that this dynamics always stops at a non-trivial equilibrium and thus in particular, such equilibria exist.
Moreover, when starting from prices of zero, the dynamics terminates at the best equilibrium.
We formalize these claims in Proposition \ref{prop:sym-dynamics} below, from which we can clearly derive the existence of non-trivial equilibrium claimed in the next theorem as an immediate corollary.



\begin{theorem}
\label{thm:existence}
		For any instance $(\vec{v},\vec{d})$ there exists at least one non-trivial pure Nash equilibrium.
\end{theorem}

Before we formally define the dynamics,
we prove a simple lemma showing that the total price weakly increases
as one seller best-responds to a higher price.



\begin{lemma}
\label{lemma:monotonicity-of-prices}
 \textbf{(Monotonicity Lemma.)}
	Let $br_x\in BR(x)$ be a best reply of a seller to a price $x$ and let $br_y\in BR(y)$ be a best reply of a seller to a price $y$.
	If $x<y\leq v_1$ then $y+br_y\geq x+br_x$.
\end{lemma}

\begin{proof}
	As $x<y\leq v_1$ by Observation \ref{obs:BR-price-is-value}, we know that there exists $i$ such that $x+br_x=v_i$ and $j$ such that $y+br_y=v_j$. 
As the second seller is best responding at each price level,
	$\mathcal{D}(v_i) (v_i - x) \ge \mathcal{D}(v_j) (v_j - x)$ and $\mathcal{D}(v_i)(v_i - y) \le \mathcal{D}(v_j) (v_j - y)$. Together, we get that $(v_j-x)/(v_i-x) \le \mathcal{D}(v_i)/\mathcal{D}(v_j) \le (v_j-y)/(v_i-y)$.
	Now notice that the function $(a-x)/(b-x)$ is non-decreasing in $x$ iff $a \ge b$ thus, since $y>x$, it
	follows that $v_j \ge v_i$.
\end{proof}

We next formally define the price-updating dynamics that we call {\em symmetrized best response dynamics}.
It works similarly to the best response dynamics with one small difference: at each step, before a seller acts, the price of both sellers is replaced by their average price.

More formally, we start from some profile of prices $(x_0,y_0)$. We then
symmetrize the prices to $(\frac{x_0+y_0}{2},\frac{x_0+y_0}{2})$, and then we let the first seller best reply to get prices $(x_1,y_1)$,
where $x_1 \in BR(\frac{x_0+y_0}{2})$ and $y_1=\frac{x_0+y_0}{2}$.
In one case, when the utility of the seller is 0, we need to break ties carefully: if  $0\in BR(\frac{x_0+y_0}{2})$ then we assume that $x_1=0$, that is, a seller with zero utility prices at 0.
We then symmetrize again to $(\frac{x_1+y_1}{2},\frac{x_1+y_1}{2})$, and then we let the second seller best respond, symmetrize again, and continue similarly in an alternating order.
The dynamic stops if the price remains unchanged in some step.

It turns out that symmetrized best response dynamics quickly converges to a non-trivial equilibrium. Moreover, we show that this dynamics is guaranteed to end up in the best equilibria. Theorem \ref{thm:existence} follows from the following proposition.

\begin{proposition}\label{prop:sym-dynamics}
For any instance
with $n$ demand levels, the symmetrized best response dynamics starting with prices $(0,0)$ reaches a non-trivial equilibrium in at most $n$ steps, in each of them the total price increases.
Moreover, 
this equilibrium achieves the highest social welfare and the highest revenue among all equilibria.


\end{proposition}
\begin{proof}
We first argue that for any starting point, the sum of players' prices in the symmetrized dynamics is either monotonically increasing or monotonically decreasing.
To see that, let us look at the symmetric price profiles of two consecutive steps: $(x,x)$ and then $(y,y)$ where
	$y=(x+br_x)/2$ for some $br_x\in BR(x)$ and then $(z,z)$ where $z=(y + br_y)/2$ for some $br_y\in BR(y)$.
If $x=y$, then $(x,x)$ is an equilibrium and we are done. 
We first observe that if $y>x$ then $z \geq y$. 
Indeed, our monotonicity lemma (Lemma \ref{lemma:monotonicity-of-prices}) shows exactly that: if $y>x$ then
for any
$br_x\in BR(x)$ and  $br_y\in BR(y)$ it holds that
 $y+br_y \geq x+br_x$ and therefore $z\geq y$. Similarly, if $y<x$ then $z \leq y$.

To prove convergence, note that until the step where the process terminates,
the total price must be either strictly increasing or strictly decreasing. Due to Observation \ref{obs:BR-price-is-value}, the total price at each step must be equal to $v_i$ for some $i$. Since there are exactly $n$ distinct values, the process converges after at most $n$ steps.
Note that if we reach a price level of $v_n$ or $v_1$ the process must stop (no seller will have a best response that crosses these values), and a non-trivial equilibrium is reached.


%
%

Finally, we will show that a symmetrized dynamics starting at zero prices reaches an equilibrium with maximal revenue and welfare over all equilibria.
Using Proposition \ref{obs:best-NE-well-defined}, it is sufficient to show that such process reaches an equilibrium with minimum total price over all possible equilibria. This follows from the following claim:

\begin{claim}
The total price reached by a symmetrized best-response dynamics starting from a total price level $x$ is bounded from above by the total price reached by the same dynamics starting from a total price of $y>x$,
%
\end{claim}
\begin{proof}
It is enough to show that the prices reached after a single step from $x$ are at most those reached by a single step from $y$, since we can then repeat and show that this holds after all future steps.  For a single step this holds due to the monotonicity lemma (Lemma \ref{lemma:monotonicity-of-prices}): given some total price $z$, the new total price after a single step of symmetrizing the price and best responding is $f(z)=z/2+br_{z/2}$ for some $br_{z/2}\in BR(z/2)$, and since $y>x$ it holds that $f(y)\geq f(x)$ by Lemma \ref{lemma:monotonicity-of-prices}.
\end{proof}

We complete the proof by showing how the proposition follows from the last claim. Let $p$ be the total price of the highest welfare equilibrium (lowest equilibrium price).
We use the claim on total price $0$ and total price $p>0$. The symmetrized best-response dynamics starting at $p$ stays fixed and the total price never changes, while the dynamics starting at $0$ must strictly increase the total price at each step, and never go over $p$, and thus must end at $p$ after at most $n$ steps.
This concludes the proof of the proposition.
\end{proof}

\section{Best Response Dynamics}
\label{sec:best-response}

In the previous section we saw that non-trivial NE always exist in our price competition model, and that the best equilibrium can be easily computed. We now turn to discuss whether we can expect agents in these markets to reach such equilibria via natural adaptive heuristics.
We consider the process of repeated best responses. Such a process starts from some profile of prices $(p,q)$, then the first seller chooses a price which is a best response to $q$, the second seller best responds to the price chosen by the first seller, and they continue in alternating order. The process stops if no seller can improve his utility by changing his price. As we aim for non-trivial equilibria, a seller that cannot gain a positive profit chooses the best response of zero.
A sequential best response process has simple and intuitive rules.
The main difference between different possible dynamics of this form is in their starting prices.
We will study the importance of the choice of starting prices.

Our results for best-response dynamics are negative: we show that starting from cartel prices might result in  bad equilibria. We then consider starting from random prices and show that this might not help.  Finally, we show that convergence time of the dynamics may be very long, even with only two demand levels.



\subsection{Quality of the Dynamics' Outcomes}

Probably the most natural starting prices to consider in best responses dynamics are $(0,0)$. We start with a simple example that shows that such dynamics might result in an equilibrium with very low welfare, even when another equilibrium with high welfare exists.
The gap between the quality of these equilibria is in the order of $D$ (in Appendix \ref{app:gaps-between-eq-tight} we show that this is the largest possible gap between equilibria).

\begin{example}
\label{obs:equal-split-2}
Consider a market with 2 demand levels, $v_1=2$, $v_2=1$, $d_1=1$ and $d_2=D$.
Here, a best response dynamics starting from prices $(0,0)$ moves to $(1,0)$ and then ends in equilibrium prices $(1,1)$. This NE has welfare of $2$,
while $(1/2,1/2)$ is an equilibrium with welfare of $D+1$ and revenue of $D$.

It follows that even with $2$ demand levels, the total revenue in the highest revenue equilibrium can be factor $D/2$ larger than both the welfare and revenue of the equilibrium reached by best-response dynamics starting from prices $(0,0)$.

\end{example}

One might hope that starting the dynamics from a different set of prices will guarantee convergence to a good equilibrium.
Clearly, if the dynamics somehow starts from the prices of the best equilibrium it will immediately stop, but our goal is exactly to study whether the agents can adaptively reach such equilibria. 
One can consider two reasonable approaches for studying the starting points of the dynamics: the first approach assumes that the sellers initially agree to act as a cartel and price the bundle at the monopolist price, dividing the monopoly profit among themselves. It is well known that such a cartel is not stable, and sellers may have  incentives to deviate to a different price; We would like to understand where such dynamics will stop.
The second approach considers starting from a random pair of prices, and hoping that there will be a sufficient mass of starting points for which the dynamics converges to a good equilibrium. We move to study the two approaches below.


\subsubsection{Dynamics Starting at a Split of the Monopolist Price }

We now study best-response dynamics that start from a cartelistic solution: the total price at the starting stage is equal to the price a monopoly would have set had it owned the two selling firms.
In Example \ref{obs:equal-split-2}  we saw that splitting the monopolist price between the two sellers results in the best  equilibrium.
One may hope that this will generalize and such starting points ensure converging to good outcomes.
In Appendix \ref{app:cartel-prices-two-levels} we show that this is indeed the case for two demand levels.
%
%
%
%
However, we next show that even with three demand levels, the welfare and revenue of the equilibrium reached by such best-response dynamics can be much lower than the revenue of the best equilibria. This holds not only when the two seller split the monopolist price evenly, but for any cartelistic split of this price. Proof can be found in Appendix \ref{app:cartel-prices}.

\begin{proposition}
	\label{obs:brd-3-lb}
	For any large enough total demand $D$ there is an instance with $3$ demand levels and monopolist price $p^*$ for which best response dynamics
	starting from any pair  $(p^*- q,q)$ for $q\in [0,p^*]$, ends in an equilibrium of welfare and revenue of only $1$, while there exist another equilibrium of welfare and revenue at least $\sqrt{D}/4$.
\end{proposition}

We conclude that starting from both sellers (arbitrarily) splitting the monopolist price does not ensure that the dynamics ends in a good equilibrium, even with only  three demand levels.

\subsubsection{Dynamics Starting at Random Prices }

We now consider a second approach for studying the role of starting prices in best-response dynamics. We assume that the starting prices are determined at random, and ask what are the chances that a sequence of best responses will reach a good equilibrium.
Unfortunately this approach fails as well. We next show that for any $\epsilon>0$, there is an instance with only two demand levels for which the dynamics starting from a uniform random price vector in $[0,v_1]^2$ has probability of at most $\epsilon$ of ending in an equilibrium with high welfare and revenue (although such equilibrium exists).\footnote{
In Appendix \ref{app:dynamic-random-prices} we show that this result is essentially tight.
}


\begin{proposition}[High probability of convergence to bad equilibria, $n=2$]
	\label{obs:conv-bad-NE-n is 2}

	For any\\   small enough $\epsilon>0$ and total demand $D$ such that $\epsilon D>2$,
	there is an instance with two demand levels ($n=2$) that has an equilibrium of welfare and revenue of at least $\epsilon D$, but best-response dynamics starting with uniform random pair of prices
in $[0,v_1]^2$ ends in an equilibrium of welfare and revenue of only $1$ with probability at least $1-\epsilon$.
\end{proposition}
\begin{proof}
	Consider the input with $n=2$ demand levels satisfying $v_1=1> v_2=\epsilon$ and $d_1=1<d_2=D$. 
	A pair of prices $(p,q)$ with $p+q=v_2$ results in welfare and total revenue of $\epsilon D$,
	and if $\epsilon D>2$, the pair $(v_2/2, v_2/2) $ is indeed an equilibrium.
	On the other hand, for small enough $\epsilon $ the pair of prices $(1/2,1/2)$ is also an equilibrium, and its welfare and revenue are only $1$. Finally, observe that unless the price that the first best response in dynamics refers to is at most $v_2=\epsilon$, the first best response results in an equilibrium with total price of $1$, and welfare and revenue of $1$. 
	The probability that the process stops after a single step is therefore at least
$1-\epsilon$, and the claim follows.
\end{proof}


Proposition \ref{obs:conv-bad-NE-n is 2} only gives high probability of convergence to a low welfare equilibrium, but this will not occur with certainty.
We next show that one can construct instances in which except of a measure zero 
set of starting prices, 
every dynamics will end up in an equilibrium with very low welfare, although equilibrium with high welfare exists.
Moreover, we show that the welfare gap between the good and bad equilibria increases exponentially in the number of demand levels $n$.

\begin{theorem}[Almost sure convergence to bad equilibria, large $n$]
	\label{thm:almost-sure-bed}
	For any 
	number of demand levels $n\geq 2$  and $\epsilon>0$ that is small enough,
	there exists an instance 
	that has an equilibrium with welfare $2\cdot (2-\epsilon)^{n-1}-1 $ 
	and revenue of $(2-\epsilon)^{n-1}$,
	but best response dynamics starting with pair of prices chosen uniformly at random over
	 $[0,v_1]^2$ almost surely ends in an equilibrium of welfare and revenue of only $1$. 
\end{theorem}

To prove the theorem, we build an instance where
the pair of prices $(v_i/2, v_i/2)$ forms an equilibrium for any $i$.
In this instance, the total revenue from a total price $v_i$ is $(2-\epsilon)^{i-1}$.
In particular, $(v_n/2, v_n/2)$ is an equilibrium that attains the monopolist revenue and the optimal welfare of $O((2-\epsilon)^n)$.
However, best response dynamics starting by best responding to any price
which is not \emph{exactly} $v_i/2$ (for some $i$)
terminates in an equilibrium with total price of $v_1=1$ and welfare of $1$. Thus, the set of pairs from which the dynamics does not end at welfare of $1$ is finite and has measure $0$, so the dynamics almost surely converges to the worst equilibrium. The full proof is in Appendix \ref{app:dynamic-random-prices-impossibility}.




\subsection{Time to Convergence}

Up to this point we considered the quality of equilibria reached by best response dynamics. In this section, we will show that not only that best response dynamics reach equilibria of poor quality, it may also take them arbitrary long time to converge.
Moreover, the long convergence time is possible even with only $2$ demand levels and total demand that is close to $1$.

Specifically, we will show that as the difference between the demand of adjacent values becomes smaller, the convergence time can increase.
More formally, we let
$W=\frac{d_n}{min_{i=2}^n \{d_i-d_{i-1}\}}$ be the ratio between the maximal demand and the minimal change in demand.
Note that if $d_1=1$ and every $d_i$ is an integer, then 
$d_1=min_{i=2}^n \{d_i-d_{i-1}\}$ and thus
$W=D$; 
if demands are not restricted to be integers, $W$ might be much larger than $D$ even in the case that $d_1=1$,
for example if $d_1=1$ and $d_2=1+\epsilon=D$
then $W=1/\epsilon$ is large while $D=1+\epsilon\approx 1$.
We show a simple setting with only two demand levels and with $D$ close to $1$ in which the dynamics takes time linear in $W$.


\begin{theorem}[Slow convergence]
	For any $W$, 
	best response dynamics starting from zero prices may require each seller to update his price $W-1$ times
	to converge to an equilibrium.
	Moreover, this holds even with 2 demand levels ($n=2$) and with $D= \frac{W}{W-1}$ which is close to $1$ when $W$ is large.
\end{theorem}
\begin{proof}
We consider the following setting given some $\epsilon>0$ that is small enough:
$v_1=1$ and $d_1=1$, $v_2=1-\epsilon$ and $d_2=\frac{1}{1-2\epsilon}$.
In this case, $W=d_2/(d_2-d_1)=\frac{1}{2\epsilon}$. We will show that for this instance best response dynamics starting at $(0,0)$
 takes at least $W-1=\frac{1}{2\epsilon}-1$ steps to converge to an equilibrium. 

Let $p_m,q_m$ denote the price offered by the two sellers after $m$ best-response steps for each seller ($p_m$ is the offer of the seller who plays first).
We will prove by induction that $p_m=1-m\epsilon$ and $q_m=m\epsilon$ whenever $m+1 < \frac{1}{2\epsilon}$.

We first handle the base case. With zero prices, the first seller can price at $v_1=1$ and get profit $1$, or price at $v_2= 1-\epsilon$ and get profit
$(1-\epsilon)\cdot \frac{1}{1-2\epsilon} >1$. Thus, $p_1=1-\epsilon$.
Now, the best response of the other seller is clearly $q_1=\epsilon$ as
pricing at total price of $1-\epsilon$ gains her $0$ profit.

We next move to the induction step. Assume that the claim is true for some $m$, i.e., $(p_m,q_m)=(1-m\epsilon, m\epsilon)$, and we prove it for $m+1$
(as long as $m+1<\frac{1}{2\epsilon}$).
If the second seller prices at $m\epsilon$, the first seller will maximize profit by pricing either at $1-(m+1)\epsilon$  or at  $1-m\epsilon$ (recall that
by Observation \ref{obs:BR-price-is-value} after a seller is best responding, the price will be equal to either $v_1$ or $v_2$).

The gain from the first price is $(1-(m+1)\epsilon)\cdot \frac{1}{1-2\epsilon}$ and the gain from the latter price is $1-m\epsilon$.
Simple algebra shows that $(1-(m+1)\epsilon)\cdot \frac{1}{1-2\epsilon} > 1-m\epsilon$ iff $m<\frac{1}{2\epsilon}$.

	Now, assume that the first seller prices at  $1-(m+1)\epsilon$,
the second seller maximizes profit by pricing either at $(m+1)\epsilon$  or at $1-\epsilon-(1-(m+1)\epsilon)=m\epsilon$.
The second seller chooses a price of $(m+1)\epsilon$ if
$(m+1)\epsilon > \frac{1}{1-2\epsilon}m\epsilon$.
Simple algebra shows that this holds iff $m+1<\frac{1}{2\epsilon}$.
This concludes the induction step and completes the proof.
\end{proof}

We observe that with two demand levels, convergence to equilibrium is guaranteed, and the above linear bound is actually tight.
Proof appears in Appendix \ref{app:BR-stops-at-time-W}.
\begin{proposition}
	\label{obs:bsd-2-stops}
	For any instance with $2$ demand levels ($n=2$), best response dynamics starting from any price profile will stop in an
	equilibrium after each seller updates his price at most $W$ times.
\end{proposition}
\section{The Quality of the Best Equilibrium}
\label{sec:pos}

In this section, we study the price of stability in our game, that is, the ratio between the quality of the best equilibrium and the optimal outcome (both for revenue and welfare).
The following theorem gives two upper bounds for the price of stability.
One bound shows that for every total demand $D$, the best equilibrium and the optimal outcome are at most factor $O(\sqrt{D})$ away, for both welfare and revenue.
The second bound is exponential in $n$, but it is independent of $D$. This implies, in particular, that the price of stability in markets with a small number of demand levels is small even for a very large $D$.

%

\begin{theorem}
	\label{thm:pos-UB}
	For any instance, the optimal welfare and the monopolist revenue are at most
	$O(\min \{2^n,\sqrt{D}\})$ times the revenue of the best equilibrium.
	
\end{theorem}
As the bound holds for the revenue of the best equilibrium, it clearly also holds for the welfare of that equilibrium.
The proof of the theorem is presented in Appendix \ref{sec:proof-pos-ub}.

The next theorem shows that the above price-of-stability bounds are tight.
It describes instances where the gap between the best equilibrium and the optimal outcome is asymptotically at least $2^n$ and $\sqrt{D}$, for both welfare and revenue.
We prove the theorem in Appendix \ref{sec:proof-pos-lb}.
%
\begin{theorem}
\label{thm:POS-LB}
For any number of demand levels $n$, there exists an instance for which
the optimal welfare and the monopolist revenue are at least factor $\Omega(2^n)$ larger than
the best equilibrium welfare and revenue, respectively.

In addition,
there exists an instance with integer demands for which
the optimal welfare and the monopolist revenue are at least factor $\Omega(\sqrt{D})$ larger than
the best equilibrium welfare and revenue, respectively.
\end{theorem}	

\ignore{ 
\begin{theorem}
	\label{thm:POS-n}
	Fix $\epsilon>0$.
	For any integer $n\geq 2$ there exist an instance $(\vec{v},\vec{d})$ with $n$ demand levels for which the best equilibrium has welfare and revenue of $1$, while the optimal welfare if at least $2^n - 1- \epsilon$ and the monopolist revenue is at least $2^{n-1}-\epsilon$.
	
	This is tight as for any instance $(\vec{v},\vec{d})$ with $n$ demand levels it holds that  the optimal welfare is at most factor $2^n -1$ larger than the welfare of the highest welfare equilibrium, and the revenue of the monopolist is at most factor $2^{n-1}$ larger than the total revenue of the highest revenue equilibrium.
\end{theorem}	
} 


%
%



\subparagraph*{Acknowledgements.}

Noam Nisan was supported by ISF grant 1435/14 administered by the Israeli Academy of Sciences and Israel-USA Bi-national Science Foundation (BSF) grant  2014389.

\bibliographystyle{plain}
\bibliography{bib}


\appendix

\section{Equilibria and Convexity}
\label{app:eq-are-convex}

\noindent \textbf{Proof of Lemma \ref{obs:total-price-convex}:}
\begin{proof}
	We assume WLOG that $p<q$. Assume for a contradiction that $(x, p+q-x)$ is not a pure NE. Then, for some $\Delta>-x$, it holds that $\mathcal{D}(p+q+\Delta)\cdot (\Delta+x)> \mathcal{D}(p+q) \cdot x$, or equivalently,
	\begin{align*}
	\mathcal{D}(p+q+\Delta)\cdot \Delta> (\mathcal{D}(p+q)- \mathcal{D}(p+q+\Delta))\cdot x.
	\end{align*}
	
	We will show that if this beneficial deviation had been to a higher (lower) price, then the same deviation would have been beneficial to the player that offered the lower (higher) price in the equilibrium $(p,q)$.
	
	If $\Delta>0$ then increasing the price by $\Delta$ is also a beneficial deviation for $p$ when the profile is $(p,q)$ and $p<q$.
	This holds since $\Delta>0$ and for the downward-sloping demand it holds that $\mathcal{D}(p+q)\geq \mathcal{D}(p+q+\Delta)$ and thus
	\begin{align*}	
	& \mathcal{D}(p+q+\Delta)\cdot \Delta\\
	& > (\mathcal{D}(p+q)- \mathcal{D}(p+q+\Delta))\cdot x \\
	& \geq  (\mathcal{D}(p+q)- \mathcal{D}(p+q+\Delta))\cdot p
	\end{align*}
	It follows that $\mathcal{D}(p+q+\Delta)\cdot (p+\Delta)> \mathcal{D}(p+q)\cdot p$ which implies that $p+\Delta$ is a beneficial deviation as claimed.
	
	If $\Delta<0$ then adding $\Delta$ to the price is also a beneficial deviation for $q$ when the profile is $(p,q)$. It holds that $\mathcal{D}(p+q)\leq \mathcal{D}(p+q+\Delta)$ ($\Delta<0$) and thus
	\begin{align*}
	& \mathcal{D}(p+q+\Delta)\cdot \Delta \\
	&> (\mathcal{D}(p+q)- \mathcal{D}(p+q+\Delta))\cdot x\\
	& \geq (\mathcal{D}(p+q)- \mathcal{D}(p+q+\Delta))\cdot q
	\end{align*}
	
	We showed that $\mathcal{D}(p+q+\Delta)\cdot (q+\Delta)> \mathcal{D}(p+q)\cdot q$ which implies that $q+\Delta$ is a beneficial deviation as claimed.
\end{proof}


\section{Quality Gaps between Equilibria}
\label{app:gaps-between-eq-tight}
\begin{proposition}\label{prop:POA-D}
	For any number of demand levels $n\geq 2$ and any total demand $D$, it holds that the ratio between the optimal welfare (and thus the welfare in the best equilibria) and the welfare in any non-trivial equilibrium is at most $D$.
	Additionally, the ratio between the revenue of a monopolist (and thus the best revenue in equilibria) and the revenue of any non-trivial equilibrium is at most $2D$.    	
\end{proposition}
\begin{proof}
	In any non-trivial NE the welfare is at least $v_1\cdot d_1$,
		while the optimal welfare is at most $d_n \cdot v_1$,
		and thus the ratio of the two is at most 
		$\frac{d_n\cdot v_1}{d_1\cdot v_1}= \frac{d_n}{d_1}=D$. 
	
	We next move to present the revenue bound. Observe that the revenue of the monopolist is at most $d_n\cdot v_1$ as this is a bound on the welfare. Fix any NE $(p,q)$ and assume wlog that $p\geq q$.
	
	If $p \mathcal{D}(p+q)\geq \frac{v_1\cdot d_1}{2}$, then the ratio between the monopolist revenue and the revenue in a non-trivial equilibrium is at most 
	$\frac{d_n\cdot v_1}{d_1\cdot v_1/2}= \frac{2d_n}{d_1}= 2D$.
	Otherwise, $(p,q)$ is an equilibrium in which $q \mathcal{D}(p+q)\leq p \mathcal{D}(p+q)< \frac{v_1\cdot d_1}{2}$ and it must be the case that $q\leq p<v_1/2$ as otherwise the revenue of the first seller is $p \mathcal{D}(p+q)\geq v_1 \mathcal{D}(p+q)/2 \geq v_1 d_1/2$. But when
$q\leq p<v_1/2$ the revenue of the second seller by pricing at $v_1-p$ is $d_1 (v_1-p) > d_1 v_1 - v_1 d_1 /2 = v_1 d_1 /2 $, a contradiction.
\end{proof}

\section{Best Response Dynamics}

\subsection{Starting from Cartel Prices}
\label{app:cartel-prices-two-levels}

\begin{proposition}
	For any instance with two demand levels ($n=2$) and any monopolist price $p^*$, best responses dynamics starting from $(p^*/2,p^*/2)$ always ends in an equilibrium with revenue that is at least half the revenue of the monopolist (and thus the revenue in any other equilibria), and welfare that is at least a third of the optimal welfare.
\end{proposition}
\begin{proof}
	Proposition \ref{obs:bsd-2-stops} shows that with two demand levels, best responses dynamics always converges to an equilibrium.
	We next prove the welfare and revenue bounds.
	
	
	Assume without loss of generality that with price of $1$ the demand is $1$, and that with price of $p<1$ the demand is $d>1$.
	If $1$ is a monopolist price, then $1\geq d\cdot p$ and $(1/2,1/2)$ is an equilibrium (since the revenue by deviation is
	$(p-1/2)d<(p/2) d\leq 1/2$) having revenue that is the same as the monopolist revenue, and welfare that is at least half the optimal welfare (optimal welfare is at most $1+d\cdot p\leq 2$).
	
	We next consider the case that $p^*=p$.
	As $p^*$ is a monopolist price it holds that $p^*\cdot d\geq 1$.
	If $p^*\cdot d>2$ then $(p^*/2,p^*/2)$ is an equilibrium with maximal revenue. Otherwise
	$2\geq p^*\cdot d\geq 1$, and the revenue in equilibrium reached by the dynamics 
	will be either $1$ or $d\cdot p^*$, and in any case, at least half the maximal equilibrium revenue.
	The welfare claim follows from the fact that for the case $p^*\cdot d>2$ then $(p^*/2,p^*/2)$ is an equilibrium of welfare at least $p^*\cdot d>2$ while the optimal welfare is at most $p^*\cdot d+1< 2d\cdot p^*$. For the case that $2\geq p^*\cdot d\geq 1$, the optimal welfare is at most $d\cdot p^* + 1\leq 3$, while any equilibrium has welfare of at least $1$.
\end{proof}

\subsection{Dynamics Starting at Any Cartelistic Split}
\label{app:cartel-prices}

Proof of Proposition \ref{obs:brd-3-lb}:
\begin{proof}
	Let $v_1=1$, $v_2= 1/4$,$v_3= 1/(3\sqrt{D}) $ and
	let $d_1=1$, $d_2= \sqrt{D}$ and $d_3=D$ (assume that $D$ is large enough).

    First observe that $(v_2/2,v_2/2)$ is an equilibrium with revenue $\sqrt{D}/4$, and at least such welfare.
    Next, observe that the monopolist price is $p^*=v_3$ as the revenue from a total price of $v_3$ is $\sqrt{D}/3$ which is greater than $\sqrt{D}/4$, which is
    the revenue with total price $v_2$. 	

    We now consider any dynamics that starts by a best response to price $q\leq v_3=p^*$. By Observation \ref{obs:BR-price-is-value}, the total price after this best response must be equal to either $v_1$, $v_2$ or $v_3$. We will handle these different cases separately:

    \vspace{2mm}

    \noindent \underline{Case 1:} The best response to $q$ is $v_1-q$.

    In this case we note that $v_1-q \geq 1-v_3=1- 1/(3\sqrt{D})$,  and the dynamics stops since when $D$ is large, this price is greater than $v_2$ and $v_3$ so the other player has no beneficial deviation.
    The welfare and revenue is $1$ as claimed.

    \vspace{2mm}

    \noindent \underline{Case 2:} The best response to $q$ is $v_2-q$.

    In this case the seller with price $q$ that is getting revenue of $q\cdot d_2 \leq v_3 d_2= 1/3$ will deviate to $1-(v_2-q)\geq 3/4$ improving his utility to at least $3/4$. Again, now the dynamics stops at equilibrium with welfare and revenue of $1$, as claimed.

    \vspace{2mm}

    \noindent \underline{Case 3:} The best response to $q$ is $v_3-q$.

    We argue that in this case the dynamics does not stop, and it must continue.
    Indeed, if for large enough $D$ the best response to $q$ is $v_3-q$ then $(v_3-q)d_3\geq (v_2-q)d_2$ or $q\leq \frac{v_3 d_3-v_2 d_2}{d_3-d_2} = \frac{\sqrt{D}}{12(D-\sqrt{D})}<\frac{1}{6\sqrt{D}}$.
    This implies that when best responding to $v_3-q$, a price of $q$ gives utility of at most $D\frac{1}{6\sqrt{D}}=\frac{\sqrt{D}}{6}$, while deviating to $q'= v_2-(v_3-q)>v_2-v_3$ ensures utility of at least $(\frac{1}{4}- \frac{1}{3\sqrt{D}}) \sqrt{D}= \frac{\sqrt{D}}{4}-\frac{1}{3}> \frac{\sqrt{D}}{6}$ (and this is clearly greater than the utility with total price $v_1$ for large $D$).
    As this seller gains at least $\frac{\sqrt{D}}{4}-\frac{1}{3}$,
    the other seller gains at most $\frac{1}{3}$ from his current price.
    However, by offering a price of $1-q'$ she can get utility of at least $\frac{3}{4}$ (since $q'>v_2=\frac{1}{4}$, note also that $q'>v_3$ so deviation to this value is not beneficial). Then the dynamics terminates as in the previous cases with revenue and welfare of 1.
\end{proof}

\subsection{Dynamics with Random Starting Prices}
\label{app:dynamic-random-prices}

We show that Proposition \ref{obs:conv-bad-NE-n is 2} is essentially tight.

\begin{proposition}
	\label{obs:conv-bad-NE-n is 2-tight}
	For any instance with two demand levels for which the ratio of welfare of the best and worst equilibrium is $\epsilon D$ for some  $1>\epsilon>2/D$, it holds that the probability of the dynamics ending at the best equilibrium when starting from a uniform random pair of prices in $[0,v_1]^2$ is at least $\epsilon-2/D$. 		
\end{proposition}
\begin{proof}
	Normalize the welfare of the worse equilibrium to $1$ (and thus the value is $1$) and the demand to $1$.
	The best equilibrium is for demand $D$ and value $\epsilon<1$, since the equilibria welfare ratio is $\epsilon D$.
	For any pair of prices $(p,q)$ such that $1/D<q<\epsilon-1/D$,
	the best response to $q$ is $\epsilon-q$ as it gives revenue larger than
	$D\cdot (1/D)=1$, while the maximal revenue for a seller in the other equilibrium is $1$.
	Given price $\epsilon-q<\epsilon - 1/D$, the best response is $q$ as it gives revenue larger than $1$, while deviation will give revenue of at most 1.
	We conclude that with probability at least $\epsilon-2/D$ the dynamics stops after a single best response, at the best equilibrium, as claimed.
\end{proof}

\subsection{Dynamics with Random Starting Prices: Impossibility}
\label{app:dynamic-random-prices-impossibility}

\noindent \textbf{Proof of Theorem \ref{thm:almost-sure-bed} (Almost sure convergence to bad equilibria)}:\\

\begin{proof}
	Let $\alpha=2-\epsilon$.
	For $i\in [n]$ let $v_i=\epsilon^{i-1} $ and $d_i=\alpha^{i-1}/v_i = \left(\frac{\alpha}{\epsilon}\right)^{i-1}$, notice that $R(v_i)=\alpha^{i-1}$. Thus, the monopolist revenue is $R(v_n)=\alpha^{n-1} = (2-\epsilon)^{n-1}$, and the optimal welfare is
	$SW(v_n) = d_1\cdot v_1+ \sum_{i=2}^n v_i (d_{i}-d_{i-1}) = 
	1+ \sum_{i=2}^n \left(\alpha^{i-1}- \epsilon \alpha^{i-2}\right) =
	1+ (\alpha - \epsilon)\cdot \sum_{i=0}^{n-2} \alpha^{i} =
	1+ (\alpha - \epsilon)\frac{\alpha^{n-1} - 1 }{\alpha -1 } = 1+ 2 (1 - \epsilon)\frac{(2-\epsilon)^{n-1}-1}{1-\epsilon} =
	1+ 2\cdot ((2-\epsilon)^{n-1}-1)  =  2\cdot (2-\epsilon)^{n-1}-1   $.
	
	We argue that for any $i\in [n]$ the pair of prices $(v_i/2, v_i/2)$ forms an equilibrium (in particular, $(v_n/2, v_n/2)$ is an equilibrium with revenue equals to the monopolist revenue, and optimal welfare),
	and that best response dynamics starting by best responding to any price $q\notin \{v_i/2 \ for \ i\in [n]\}$, ends in an equilibrium with total price of $v_1=1$, and welfare of $1$. Thus, the set of pairs from which the dynamics does not end at welfare of $1$ is finite, and has measure $0$, so the dynamics almost surely converges to total price of $1$ and welfare of $1$.
	
	We first observe that for any $i\in [n]$ prices $(v_i/2, v_i/2)$ form an equilibrium.
	Note that for $\epsilon$ that is small enough, $v_i/2>v_j$ for any $j>i$, and thus a best response to any price of at least $v_i/2$ must be $v_j-(v_i/2)$ for some $j\leq i$.
	Now, with prices $(v_i/2, v_i/2)$ the utility of each seller is
	$d_i \cdot v_i/2 = \alpha^{i-1}/2$, while any optimal deviation must be to some price $v_j - (v_i/2)$ for $j<i$ and it
	gives utility of $d_j (v_j - (v_i/2)) < d_j \cdot v_j= \alpha^{j-1} \leq  \alpha^{i-1}/2$.
	
	We next show that best response dynamics starting by best responding to any price $q\notin S= \{v_i/2 \ for \ i\in [n]\} $, ends in an equilibrium with total price of $v_1=1$, and welfare of $1$.
	Let $p\in BR(q)$ and let $i$ be the index such that  $p+q=v_i$. Since $q\notin S$ it holds that $p\neq q$.
	We argue that from this point onwards, unless the total price is $v_1$, the dynamics continues and the total price strictly increases at every best responses, thus ending at $v_1$ after at most $n$ steps.

	We first show that any uneven split is not an equilibrium. The seller with the low price will want to move to a higher total price, and his new price will be larger than the price of the other seller.
	\begin{lemma}
		\label{lem:non-equal-not-ne}
		Any pair $(x,y)$ such that $x+y=v_j<v_1$ and $x\neq y$ is not an equilibrium.
		Moreover, for small enough $\epsilon$, for $x<\frac{v_j}{2}<y$, it holds that for any $z\in BR(y)$
		we have
		$y+z=v_k>v_j$ for some $k<j$, and $z>y$.
	\end{lemma}
	\begin{proof}
		Assume that $x<y$ (thus $y>v_j/2$) we show that $v_{j-1}-y$ is better response than $v_j-y$ for $j>1$.
		Indeed $(v_{j-1}-y)d_{j-1}> (v_{j}-y)d_{j}$ since
		\begin{align*}	
		& (v_{j-1}-y)d_{j-1}> (v_{j}-y)d_{j} \Leftrightarrow  y(d_j-d_{j-1})> \alpha^{j-1} - \alpha^{j-2} \Leftrightarrow\\
		& y\left( \left(\frac{\alpha}{\epsilon}\right)^{j-1} - \left(\frac{\alpha}{\epsilon}\right)^{j-2}  \right) > \alpha^{j-1} - \alpha^{j-2} \Leftrightarrow  y>\left(\frac{\alpha-1}{\alpha-\epsilon}\right) v_j = \frac{v_j}{2}
		\end{align*}
		Finally, note that for small enough $\epsilon$ it holds that $v_j/2>v_{j+1}$ and thus $y>v_{j+1}$, so
		$z$ that is a best response to $y$ must
		satisfy $y+z\geq v_j$, and as we saw that $y+z\neq v_j$ we conclude that $y+z=v_k>v_j$ for some $k<j$.
		For small enough $\epsilon$ it holds that $z = v_k-y\geq (v_j/\epsilon)- y \geq (v_j/\epsilon) - v_j > v_j \geq  y$,
		thus $z>y$ as claimed.
	\end{proof}
	
 	We next show that in any uneven split, the seller with the higher price is best responding.
	\begin{lemma}
		\label{lem:non-equal-high-stays}
		Assume that $\epsilon>0$ is small enough.
		If for some $j<n$ it holds that $v_{j+1}<x<v_j/2$, then there is a unique best response to $x$ and it holds that  $BR(x)=\{v_j-x\}$.
	\end{lemma}
	\begin{proof}
		Since $x>v_{j+1}$ it holds that $BR(x)=\{v_k-x\}$ for some $k\leq j$.
		To prove the claim we show that for $x<v_j/2$, for any $k<j$ it holds that $(v_{j}-x)d_{j}> (v_{k}-x)d_{k}$.
		Let $m=j-k$ and note that $d_j=d_k \cdot \left(\frac{\alpha}{\epsilon}\right)^{m}$. It holds that:
		\begin{align*}	
		& (v_{j}-x)d_{j}> (v_{k}-x)d_{k} \Leftrightarrow  \\
		& \alpha^{j-1} - \alpha^{k-1} >x(d_j-d_k)\Leftrightarrow\\
		& \alpha^{k-1} (\alpha^m -1 )> x\cdot \alpha^{k-1} \left(\frac{\alpha^m-\epsilon^m}{\epsilon^{j-1}}\right)\Leftrightarrow\\
		& x< \epsilon^{j-1} \frac{\alpha^m -1}{\alpha^m -\epsilon^m} 
		\end{align*}
		Observe that for $m=1$ it holds that $\frac{\alpha^m -1}{\alpha^m -\epsilon^m}=\frac{1}{2}$. It is easy to check that
		$\frac{\alpha^m -1}{\alpha^m -\epsilon^m}$ is increasing in $m$. Thus, if $x<\frac{v_j}{2}$ then
		$(v_{j}-x)d_{j}> (v_{k}-x)d_{k}$ for any $j$ such that $j-k\geq 1$ (any $j>k$).
	\end{proof}
	
	We now prove the theorem using these two lemmas.
	Recall that $p\in BR(q)$ and assume that $p+q=v_i<v_1$. As $q\notin S$, it holds that $q\neq v_i/2$. By Lemma \ref{lem:non-equal-not-ne},
	$(p,q)$ is an uneven split and thus not an equilibrium, so $q\notin BR(p)$. When $\epsilon$ is small enough, by Lemma \ref{lem:non-equal-high-stays} it must be the case that $p\geq v_i/2$ or $p\leq v_{i+1}$.
	By Lemma \ref{lem:very uneven split} it cannot be the case that $p\leq v_{i+1}$.
	Thus, it must hold that $p\geq v_i/2$ and then $q<v_i/2$ since $p\neq v_i/2$. We can now use Lemma \ref{lem:non-equal-not-ne} inductively, to conclude that the dynamics can only stop when the total price is $v_1$.
	
	\begin{lemma}
		\label{lem:very uneven split}
		Assume that $\epsilon>0$ is small enough.
		If $p+q=v_i<v_1$ and $p\in BR(q)$ then $p>v_{i+1}$.
	\end{lemma}  	
	\begin{proof}
		Assume in contradiction that $p\leq v_{i+1}$ and thus $q=v_i-p\geq v_i-v_{i+1}$.
		The revenue of the seller with price $p$ is
		$p\cdot d_i\leq v_{i+1}\cdot d_i = \epsilon^ i \cdot \left(\frac{\alpha}{\epsilon}\right)^{i-1} = \epsilon \alpha^{i-1} $.
		On the other hand, if it response to $q$ the seller prices at $v_{i-1}-q$, his revenue is $(v_{i-1}-q) d_{i-1} \geq  \alpha^{i-2} - (v_i-v_{i+1})\cdot d_{i-1} =
		\alpha^{i-2} - (\epsilon^{i-1}-\epsilon^i)\cdot \left(\frac{\alpha}{\epsilon}\right)^{i-2} =
		\alpha^{i-2} \left(1- \epsilon +\epsilon^2 \right) $.
		Observe that we get a contradiction when $\epsilon$ is small enough, as
		$\epsilon \alpha^{i-1} \geq 	\alpha^{i-2} \left(1- \epsilon +\epsilon^2 \right) $ implies that
		$\alpha \geq \frac{1- \epsilon +\epsilon^2}{\epsilon}$, but $\alpha<2$ while the RHS goes to infinity when $\epsilon$ goes to $0$.
	\end{proof}
	
	This concludes the proof of the theorem.
\end{proof}

\subsection{Time to Covergence}
\label{app:BR-stops-at-time-W}

\noindent \textbf{Proof of Proposition \ref{obs:bsd-2-stops}}:
\begin{proof}
	We prove the result for every market with two demand levels.
	Let $v_1=1$ and $d_1=1$ (we normalize the two values to 1 w.l.o.g.), and let $v_2=1-\epsilon$ for $\epsilon>0$ and $d_2=D$.
	For these parameters, $W=\frac{D}{D-1}$.
	
	We will first show that as long as the best response process proceeds,
	there is an increase of exactly $\epsilon$ between any two consecutive prices one seller sets, and a decrease of exactly $\epsilon$ between prices set by the other seller.
	
	Recall that due to Observation \ref{obs:BR-price-is-value}, after a seller is best responding, the price will be either $1$ or $1-\epsilon$. 

	Consider first a set of prices $(p,\mathbf{q})$ where $q+p=1$ (the price that is marked in bold indicates the price of the player whose turn is to best respond, in this case, the second seller).
	The seller sets his price to a price in $BR(p)$,
	and for the dynamics to continue it must hold that the total price is now equal to $v_2$:
	$BR(p)$ is unique and not equal to $q$ and it satisfies $BR(p)+p=1-\epsilon$. Since $p+q=1$, we get that
	$BR(p)=q-\epsilon$ and the new pair of offers is $(\mathbf{p},q-\epsilon)$.
	
	Consider now some set of prices $(\mathbf{p'},q')$ where $q'+p'=1-\epsilon$.
	The seller sets his price to a price in $BR(q')$,
	and for the dynamics to continue it must hold that $BR(q')$ is unique and not equal to $p'$ and it satisfies
	$BR(q')+q'=1$. Since $q'+p'=1-\epsilon$, we get that the new pair of prices is $(p'+\epsilon,\mathbf{q'})$.
	
	We conclude that every best response dynamics have the following form.
	After the first step, the sum of prices will either $v_1$ or $v_2$.
	As long as the process proceeds, we will have the following sequence of prices when the initial total price is $v_1$ (otherwise, consider the sequence starting from the second price vector):
	$(p_0,\pmb{q_0)}$,
	$(\pmb{p_0},q_0-\epsilon)$,
	$(p_0+\epsilon,\pmb{q_0-\epsilon})$,
	$(\pmb{p_0+\epsilon},q_0-2\epsilon)$,$(p_0+2\epsilon,\pmb{q_0-2\epsilon})$,
	$...$,$(p_0+m\epsilon,\pmb{q_0-m\epsilon})$,
	$(\pmb{p_0+m\epsilon},q_0-(m+1)\epsilon)$, and so on.
	
	As prices are bounded in $[0,1]$, the number of updates by one seller clearly cannot be more than $1/\epsilon$.
	
	We are left to bound the number of iterations as a function of $D$.
	Indeed, consider the price vector $(\pmb{p_0+m\epsilon},q_0-(m+1)\epsilon)$.
	For the dynamics to continue, the currently responding player must prefer
	increasing his price and selling to the lower demand at price $v_1$:
	\begin{align}
	1\cdot(p_0+(m+1)\epsilon) > D\cdot(p_0+m\epsilon)
	\end{align}
	It follows that $m<\frac{1}{D-1}-\frac{p_0}{\epsilon}<\frac{D}{D-1}=W$. Therefore, in every best response dynamics each player will change its price at most $W$ times.
	%
	%
\end{proof}

\subsection{Proof of Theorem \ref{thm:pos-UB}}
\label{sec:proof-pos-ub}

We start with the upper bounds of the form $O(\sqrt{D})$.
The following lemmas will be useful in proving the upper bounds presented in Theorem \ref{thm:pos-UB}.
\begin{lemma}\label{lem-sqrt-bound}
	Assume that the best reply
	to $v/2$ is $v'-v/2$ for some $v'>v$.
	Then $v\leq v' \sqrt{\frac{D(v')}{D(v)}}$
\end{lemma}
\begin{proof}
 Denote $d = \D(v)$, $d'=\D(v')$ and the revenues by $r=v\cdot d$, $r'=v'\cdot d'$.
 To prove the claim we show the equivalent that $(r'/r)^2 \ge d'/d$.
	
	Denote $\alpha = d'/d$.  Since $v'-v/2$ is a better reply to $v/2$ than
	$v/2$ is we have that $d\cdot v/2 < d' (v'-v/2)$, equivalently
	$v(d+d')/2 < d'\cdot v'$ so $v(1+\alpha)/2 < \alpha v'$ or
	$v'/v > (1+\alpha)/(2\alpha)$.
	
	Now, $(r'/r)^2= (\alpha v'/v)^2 >
	((1+\alpha)/2)^2 \ge \alpha = d'/d$, where the first inequality was
	just derived above, and the last one holds for any real number $\alpha$.
\end{proof}

\ignore{	
We first show that the revenue of the best equilibrium is at most  $\sqrt{D}$ factor away from the monopolist revenue.
\begin{lemma}
	The best equilibrium has revenue of at least $1/\sqrt{D}$ fraction
	of the monopolist revenue.
\end{lemma}
\begin{proof}
	Let us consider the symmetrized best-reply dynamics starting from
	an equal split of the minimal monopoly price $p_0$.
	This gives	us a sequence $p_0 < p_1 < \cdots < p_t$ of total prices, where
	at each stage $p_{i+1}-p_i/2$ is a best response to $p_i/2$, and
	$(p_t/2,p_t/2)$ is the best equilibrium.  Denote the
	demand at combined price $p_i$ by $d_i = \D(p_i)$, and the
	revenue by $r_i = d_i \cdot p_i$.
	
	We can now apply the previous lemma to each stage and get
	$(r_{i+1}/r_i)^2 \ge d_{i+1}/d_i$, and putting all these inequalities
	together get $(r_t/r_0)^2 \ge d_t/d_0$, or $r_t \ge r_0 \sqrt{d_t/d_0}$.  The theorem follows since
	$d_t/d_0 \ge 1/D$, $r_t$ is the revenue in the best equilibrium and $r_0$ is the monopolist revenue.
\end{proof}
} 

We next show that the revenue of the best equilibrium is at most  $O(\log D)$ factor away from its welfare.
\begin{lemma} \label{revwel}
	Let $(v^*/2,v^*/2)$ be an equilibrium then the welfare at this equilibrium is at most $O(\log D)$ times the
	revenue at this equilibrium.
\end{lemma}
\begin{proof}
	Let $d^*=\D(v^*)$ be the demand at this equilibrium, then
	$v^*/2$ is at least as good a reply to $v^*/2$ as is $v-v^*/2$ (for any $v$).  For $v > v^*$, the revenue
	from the latter choice is at least $\D(v) \cdot v/2$ and thus for every $v \ge v^*$ we have that
	$v \cdot \D(v) \le v^* \cdot d^*$.
	In particular, for $v_i \ge v^*$, since $d_i \ge d^* / D$, we also have that $v_i \le D v^*$.
	
	We now split the welfare that is obtained in this equilibrium
	into parts according to the value: $S_j = \{ i | 2^j v^* \le v_i < 2^{j+1} v^* \}$,
	where $j = 0 ... \log D$.  The welfare that we get from all buyers whose value is in $S_j$
	can be bounded from above by $2^{j+1} v^* \cdot \D(2^j \cdot v^*) \le 2 v^* \cdot d^*$
	by applying the inequality $v \cdot \D(v) \le v^* \cdot d^*$ for $v=2^j \cdot v^*$.
	This completes the proof of the lemma.
\end{proof}

The upper bounds in the theorem follows from the next propositions.


We start with a preliminary proposition that shows that the revenue of
the best equilibrium is $\sqrt{D}$-competitive with the monopolist {\em revenue}.
Below we will strengthen this proposition showing that it is even competitive with respect
to the optimal social welfare.

\begin{proposition}
The best equilibrium has revenue of at least $1/\sqrt{D}$ fraction
of the monopolist revenue.
\end{proposition}

\begin{proof}
Let us consider the symmetrized best-reply dynamics starting from
an equal split of the monopoly price.  By Proposition \ref{prop:sym-dynamics} and its proof, this gives
us a sequence $p_0 < p_1 < \cdots < p_t$ of price levels where
at each stage $p_{i+1}-p_i/2$ is a best response to $p_i/2$, and
$(p_t/2,p_t/2)$ is the best equilibrium.  Denote the
demand at combined price $p_i$ by $d_i = \D(p_i)$, and the
revenue by $r_i = d_i p_i$.

We can now apply the previous lemma to each stage and get
$(r_{i+1}/r_i)^2 \ge d_{i+1}/d_i$, and putting all these inequalities
together get $(r_t/r_0)^2 \ge d_t/d_0$.  The theorem follows since
$d_t/d_0 \ge 1/D$ and $r_0$ is the monopolist revenue.
\end{proof}

We now provide the stronger proof.

\begin{proposition}
	The best equilibrium has revenue of at least $\Omega(1/\sqrt{D})$
	fraction of the optimal social welfare.
\end{proposition}

\begin{proof}
For ease of notation we wish to convert the setting so that the demand is given as $k$ unit-demands
at prices $v_1 > v_2 > \cdots v_k$, so that for every $i$,
$\D(v_i)=i$.  The way that we do this is by first multiplying the demands (that we assume are rational numbers\footnote{
	If the demands are irrational then we $\epsilon$-approximate them by rational numbers without changing the internal
	order between any two $(v_i - v_j/2) \cdot \D(v_i - v_j/2)$ thus maintaining the symmetric equilibria exactly and the approximation factors to
	within an error that can go to 0 in the limit.})
by the common denominator and then replacing the multiple units of demand at each price by multiple $\epsilon$-perturbed
values with strict inequalities between these values, maintaining exactly the demand at each original price.
Notice that this transformation maintains (exactly) the equilibria of the original game.
The parameter $D$ of the original setting
is now read as $D=k/h$, where $h$ is the index in the new game that corresponds to the lowest possible non-zero demand
in the old game.
Let us further denote the index of the best equilibrium by $t$
so that $(v_t/2,v_t/2)$ is an equilibrium with demand $t$ (and revenue
$r_t=tv_t$).  Clearly $t \ge h$, so $k/t \le D$.

If we apply the previous symmetrized best-response
process, as in the proof of the previous proposition, but now starting from
any price $p_0 = v_i$ with $i \ge t$.  Again, by Proposition \ref{prop:sym-dynamics} and its proof
we always reach the same (best) equilibrium
point $(v_t/2, v_t/2)$ so using the same argument
we have $(t v_t / i v_i)^2 \ge t/i$ and thus $v_i \le \sqrt{t/i} \cdot v_t$.

We will calculate the total welfare in two parts: from
players whose value is greater than $v_t$ and from those whose
value is less or equal to $v_t$.  The latter is
$\sum_{i=t}^k v_i \le \sqrt{t} \cdot v_t \cdot \sum_{i=t}^k (1/\sqrt{i}) \le
\sqrt{t} \cdot v_t \cdot O(\sqrt{k}) \le t \cdot v_t \cdot O(\sqrt{k/t}) = O(\sqrt{D} \cdot r_t)$.
Lemma \ref{revwel} states that the former,
$\sum_{i=1}^t v_i$ is bounded from above by $O(\log D \cdot r_t)$.
Adding the two parts up concludes
the proof.
\end{proof}

\begin{proposition}
	For any instance $(\vec{v},\vec{d})$ with $n$ demand levels it holds that  the optimal welfare is at most factor $2^n -1$ larger than the welfare of the highest welfare equilibrium, and the revenue of the monopolist is at most factor $2^{n-1}$ larger than the total revenue of the highest revenue equilibrium.
\end{proposition}
\begin{proof}
	Consider the symmetrized best response dynamics starting from $(0,0)$. After the first best responses the total price equals the price of the monopolist $p^*$. From this point the dynamics continues, and we have shown that is stops at the best equilibrium (one with the highest welfare and revenue among all equilibria), increasing the total price at each step of the dynamics, till the dynamics stops.
	We claim that the $i$ steps (increases to the total price) the revenue is at least $R(p^*)/2^i$. This is so as at each step each of the sellers gets half the revenue of the given price, and for the dynamics not to stop, the revenue he gets after increasing his price must be at least as before, thus the new total is at least half the previous total revenue. We conclude by induction that after the maximal number of prices increases from the monopolist price (at most $n-1$ increases), the total revenue is at least $R(p^*)/2^{n-1}$.
	
	
	We saw that the total revenue (and thus welfare) in equilibrium is at least  $R(p^*)/2^{n-1}$. 
	The ratio between the best equilibrium welfare and the optimal welfare is therefore at least
	$$\frac{R(p^*)}{SW(v_n)\cdot 2^{n-1}}\geq \frac{R(p^*)}{\left(\sum_{i=1}^n (v_i \cdot d_{i})\right)\cdot 2^{n-1}}=
	\frac{1}{\left(\sum_{i=1}^n \frac{v_i \cdot d_{i}}{R(p^*)}\right)\cdot 2^{n-1}}\geq $$
	$$\frac{1}{\left(\sum_{i=0}^{n-1} \frac{1}{2^i}\right)\cdot 2^{n-1}}\geq
	\frac{1}{\left(2-\frac{1}{2^{n-1}}\right)\cdot 2^{n-1}}=
	\frac{1}{2^{n}-1}$$
	as claimed.
\end{proof}

\subsection{Proof of Theorem \ref{thm:POS-LB}}
\label{sec:proof-pos-lb}

We first present the lower bound that grows asymptotically as $\sqrt{D}$.

\begin{lemma}
	\label{lem:POS-LB-D}
	For any integer $D>2$ there exist an instance $(\vec{v},\vec{d})$ with total demand $D$ (and $n=D$)
	and demands that are integer multiple of $d_1=1$ for which
	welfare and revenue in equilibrium is at most $3/\sqrt{2}$ while the optimal welfare is at least $\sqrt{D}$ and
	the revenue of a monopolist is  at least $D/\sqrt{D-1}= \Theta(\sqrt{D})$.
\end{lemma}	
\begin{proof}
	Let $v_1= 1.001$ and let $v_i = 1/\sqrt{i-1}$ for $i\in \{2,4,\ldots,n  \}$. Let $d_i=i$ and note that $D=n$.
	We show that in any equilibrium the welfare (and thus the revenue) is at most $3/\sqrt{2}$,
	while the revenue of a monopolist in at least $D/\sqrt{D-1}$ (by pricing at $1/\sqrt{D-1}$), and the social welfare is
	$\sum_i v_i = 1.001+\sum_{i=2}^n 1/\sqrt{i-1} \geq \sqrt{D}$.
	
	We next show that in any equilibrium the welfare (and thus the revenue) is at most $3/\sqrt{2}$.
	Assume that we have an equilibrium $(p,q)$ such that $p+q=v_i$ and $i\geq 3$.
	By Lemma \ref{obs:total-price-convex} $(v_i/2, v_i/2)$ must also be an equilibrium.
	Since deviating to $v_{i-2}-v_i/2$ is not profitable for a player, we have that
	$iv_i/2 \ge (i-2)(v_{i-2}-v_i/2)$, i.e. $(i-1)v_i \ge (i-2)v_{i-2}$ so,
	unless $i \le 3$ (where $v_{i-2}=v_1=1.001$
	is not given by $1/\sqrt{i-1}$), we have $(i-1)/\sqrt{i-1} \ge (i-2)/\sqrt{i-3}$, which is false since
	$(i-1)(i-3) < (i-2)^2$. It follows that the only equilibria have $i \le 3$ for which the revenue and
	welfare are at most  $3/\sqrt{2}$.
\end{proof}

We next present the lower bound that  grows exponentially in $n$.

\begin{lemma}
	\label{lem:POS-LB-n}
	Fix $\epsilon>0$.
	For any integer $n\geq 2$ there exist an instance $(\vec{v},\vec{d})$ with $n$ demand levels for which the best equilibrium has welfare and revenue of $1$, while the optimal welfare is at least $2^n - 1- \epsilon$ and the monopolist revenue is at least $2^{n-1}-\epsilon$.
\end{lemma}
\begin{proof}
	Let $\delta>0$ be small enough.
	Let $\alpha=2-\delta$.
	For $i\in [n]$ let $v_i=\delta^{i-1} $ and $d_i=\left(\alpha^{i-1} - \delta^{n-i+1}\right)/v_i$.
	Note that $R(v_i) = v_i\cdot d_i = \alpha^{i-1} - \delta^{n-i+1}$ and for small enough $\delta$ the monopolist price is $v_n$
	and the monopolist revenue is $R(v_n) = v_n\cdot d_n = \alpha^{n-1} -\delta =  (2-\delta)^{n-1} - \delta $
	which tends to $2^{n-1}$ as $\delta$ goes to $0$.
	The welfare at the monopolist price $v_n$ is
	$SW(v_n) = d_1\cdot v_1+ \sum_{i=2}^n v_i (d_{i}-d_{i-1}) = 
	1+ \sum_{i=2}^n \left(\alpha^{i-1}- \delta \alpha^{i-2} - \delta^{n-i+1} + \delta^{n-i+3} \right)=   2\cdot (2-\delta)^{n-1}-1 +f(\delta)$,
	when $f(\delta)$ is a function that tends to $0$ as $\delta$ tends to $0$, and this welfare tends to
	$2^n - 1$ when $\delta$ goes to $0$.

	We argue that for any equilibrium $(p,q)$ it holds that $p+q=v_1 =1$, and that there is at least one such equilibrium (this is trivial at $(1/2,1/2)$ is clearly an equilibrium when $\delta<1/2$).
	
	To prove that there is no equilibrium $(p,q)$ such that  $p+q=v_i<v_1$ it is enough to prove that $(v_i/2,v_i/2)$ is not an equilibrium for any $i>1$ (by Lemma \ref{obs:total-price-convex}). Indeed, we show that $v_i/2$ is not a best response to $v_i/2$ by showing that $v_{i-1}-v_i/2$ gives higher revenue:
	\begin{align*}	
	& d_i\cdot v_i/2<d_{i-1}(v_{i-1}-v_i/2) = d_{i-1}\cdot v_{i-1}(1-\delta/2)  = d_{i-1}\cdot v_{i-1}\cdot (2-\delta)/2  \\
	& \Leftrightarrow \alpha^{i-1} - \delta^{n-i+1} < \left(\alpha^{i-2} - \delta^{n-i+2}\right) \cdot \alpha \Leftrightarrow  1>\alpha\cdot \delta = (2-\delta)\cdot \delta
	\end{align*}
	which holds when $\delta$ is small enough.	 	
\end{proof}

\end{document}